\documentclass[acmsmall,nonacm]{acmart}
\usepackage{algorithm}
\usepackage[noend]{algorithmic}
\usepackage{mathtools}
\usepackage[skip=0pt]{caption}

\newcommand{\FragSet}{F}
\newcommand{\HashSet}{H}
\newcommand{\ReceiveSet}{R}
\newcommand{\ProposeSet}{P}
\newcommand{\SigSet}{S}
\newcommand{\sender}{sender}

\begin{document}

\sloppy

\title{Byzantine Reliable Broadcast with Low Communication and Time Complexity}

\author{Thomas Locher}
\affiliation{%
  \institution{DFINITY}
  \city{Zurich}
  \country{Switzerland}
}
\email{thomas.locher@dfinity.org}

\renewcommand{\shortauthors}{Locher}

\begin{abstract}
Byzantine reliable broadcast is a fundamental problem in distributed computing, which has been studied extensively over the past decades. State-of-the-art algorithms are predominantly based on the approach to share encoded fragments of the broadcast message, yielding an asymptotically optimal communication complexity
when the message size exceeds the network size, a condition frequently encountered in practice.
However, algorithms following the standard coding approach incur an overhead factor of at least 3, which can already be a burden for bandwidth-constrained applications. Minimizing this overhead is an important objective with immediate benefits to protocols that use a reliable broadcast routine as a building block.

This paper introduces a novel mechanism to lower the communication and computational complexity. Two algorithms are presented that employ this mechanism to reliably broadcast messages in an asynchronous network where less than a third of all nodes are Byzantine. The first algorithm reduces the overhead factor to $2$ and has a time complexity of $3$ if the sender is honest, whereas the second algorithm attains an optimal time complexity of $2$ with the same overhead factor in the absence of equivocation.
Moreover, an optimization is proposed that reduces the overhead factor to $3/2$ under normal operation in practice. Lastly, a lower bound is proved that an overhead factor lower than $3/2$ cannot be achieved for a relevant class of reliable broadcast algorithms.
\end{abstract}

\ccsdesc[500]{Computer systems organization~Fault-tolerant network topologies}
\ccsdesc[500]{Theory of computation~Distributed algorithms}

\keywords{asynchronous networks, reliable broadcast, communication complexity}

\maketitle

\section{Introduction}
\label{sec:introduction}

The goal of Byzantine reliable broadcast is to disseminate a message efficiently and reliably despite the presence of Byzantine nodes that may interfere with the protocol execution in arbitrary ways.
A reliable broadcast routine is a powerful primitive with a broad range of applications including asynchronous atomic broadcast~\cite{duan2018, keidar2021, guo2020, miller2016, zhang2022}, distributed key generation~\cite{abraham2021, das2022, kokoris2020}, secure data replication~\cite{cachin2002}, and secret sharing~\cite{yurek2021}.
Moreover, Byzantine reliable broadcast plays a pivotal role in Byzantine fault-tolerant consensus protocols where message dissemination is treated separately from message ordering to improve throughput~\cite{danezis2022}.

The first Byzantine reliable broadcast algorithm is due to Bracha~\cite{bracha1987}. While it is elegant in its simplicity, its main drawback is that the message $m$ is broadcast by every node during the execution of the algorithm, i.e., $O(|m|n^2)$ bits are sent overall, where $|m|$ denotes the size of $m$ in bits and $n$ is the number of nodes in the network.
In the seminal paper by Cachin and Tessaro~\cite{cachin2005}, this bound was improved to
$O(|m|n + \kappa n^2\log(n))$ bits using \emph{erasure coding}, where $\kappa$ is the output size of a collision-resistant hash function. Since each node must receive the message in a successful broadcast, this result is asymptotically optimal if $|m| \in \Omega(\kappa n\log(n))$. Subsequent work has primarily focused on getting closer to the lower bound of $\Omega(|m|n + n^2)$~\cite{abraham2022, alhaddad2022, das2021, nayak2020, patra2011}.
These pieces of work follow the blueprint laid out by Cachin and Tessaro and augment it with error-correction and cryptographic primitives to achieve better asymptotic bounds.

Curiously, there is little work on minimizing the constant in the $O(|m|n)$ term, although it likely dominates the actual bandwidth consumption in practice and is therefore crucial for real-world applications. The design by Cachin and Tessaro is based on a $(n, t+1)$-erasure code, where $t < n/3$ is the largest number of nodes that may exhibit Byzantine behavior.
In this design, nodes broadcast encoded \emph{fragments} of size $|m|/(t+1) \approx 3|m|/n$  instead of broadcasting $m$, sending at least $3|m|n$ bits altogether. An overhead factor of $3$ -- compared to the ideal scenario where each node receives exactly $|m|$ bits -- may already prove too costly for applications with strict bandwidth constraints.
All schemes that follow this approach naturally inherit the overhead inherent in this design.

In this paper, a mechanism is introduced that departs from this blueprint in order to reduce the communication complexity, i.e., the number of bits that need to be exchanged in the worst case, for large messages. The core idea is quite simple: a $(n, 2t+1)$-erasure code is used instead, reducing the fragment size to $|m|/(2t+1) \approx \frac{3}{2} |m|/n$.
This modification obviously reduces the cost of broadcasting fragments; however, correctness is no longer given without additional changes because some nodes may receive $2t+1$ fragments whereas others merely obtain $t+1$ fragment, which is not enough to recover the message.
This problem is addressed by introducing the following step. If a node with $2t+1$ fragments manages to reconstruct the message, it disseminates fragments again but only to the $t$ nodes from which no fragment was received. As we will see, this step is sufficient to ensure correctness while keeping the communication complexity low.

\begin{figure}[t]
\center
\includegraphics[width=\columnwidth]{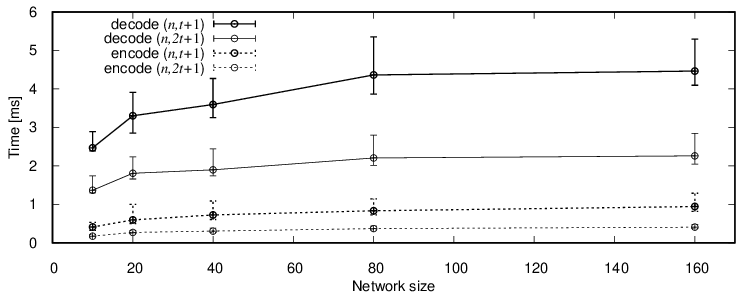}
\caption{Given an input of size 1 MiB, encoding and decoding times for erasure codes with parameters $(n, t+1)$ and $(n, 2t+1)$ are shown for different network sizes. The
circles and bars mark the averages and the 5th and 95th percentiles, respectively, of the measured times over 1000 runs.}
\label{fig:ec_performance}
\end{figure}

As an added benefit, the transition from a $(n, t+1)$-erasure code to a $(n, 2t+1)$-erasure code results in substantial performance improvements, both with respect to the encoding and decoding of fragments. Figure~\ref{fig:ec_performance} shows performance numbers when encoding/decoding an input of size 1 MiB for various network sizes and the two different parameterizations. The experiments were conducted on an Intel Core i7 CPU with 6 cores and 32 GB of memory using an optimized library\footnote{See \url{https://github.com/AndersTrier/reed-solomon-simd}.}
for Reed-Solomon codes~\cite{reed1960}.
The figure confirms the practicality of erasure coding due to its efficiency, exhibiting encoding and decoding times in the millisecond range, and scalability with respect to the network size.
Moreover, the figure reveals that the encoding (decoding) time is lower by a factor of 2.2--2.4 (1.8--2.0) for all network sizes when using an $(n, 2t+1)$-erasure code. Another interesting property is that the variance is also significantly lower, leading to more predictable encoding and decoding times.

The paper is organized as follows. The model is formally introduced in \S\ref{sec:model}.
A concrete algorithm that makes use of the novel mechanism is presented in \S\ref{sec:algorithm}. Its communication complexity is $2|m|n$ plus a term independent of $|m|$. Given an \emph{honest} (i.e., non-faulty) sender, each honest node obtains the message after $3$ rounds of communication.
A natural question is whether there is an algorithm with an overhead factor lower than $3$ and that terminates in an optimal $2$ rounds given an honest sender.
In \S\ref{sec:variant}, this question is answered in the affirmative by presenting an algorithm with this property and the same overhead factor if the sender does not equivocate. Otherwise, the worst-case overhead factor is $\frac{5}{2}$ with respect to the largest fragment size transmitted by honest nodes.
While the first algorithm only requires collision-resistant hash functions, the second algorithm makes use of threshold signatures.
Both algorithms further bound the number of bits that are stored at each node in the worst case to a small multiple of a predetermined maximum message size.
Practical considerations with the goal of reducing the communication complexity further in real-world deployments and simplifying the implementation are discussed in \S\ref{sec:optimizations}. A key observation is that the overhead can be reduced to $\frac{3}{2}$ during periods of synchrony and in the absence of failures.
As shown in \S\ref{sec:lower_bound}, a lower overhead factor than $\frac{3}{2}$ cannot be guaranteed for a particular class of algorithms.
Related work is summarized in \S\ref{sec:related_work} before the paper concludes in \S\ref{sec:conclusion}.

\section{Model}
\label{sec:model}

The considered network comprises $n = 3t+1$ nodes, where $t$ denotes the maximum number of \emph{Byzantine nodes}, which may deviate arbitrarily from any given protocol.
The other $2t+1$ nodes faithfully execute the given protocol and are called \emph{honest}.
Any two nodes can communicate directly over an authenticated channel by exchanging messages.
Communication is assumed to be \emph{asynchronous} in the sense that the delivery of messages, while guaranteed, may be delayed indefinitely. In this communication model, nodes cannot make any assumption about message delays, neither about the time required until sent message arrive at their destinations nor about the time that elapses until certain (expected) messages are received.

A node may send a message to a single node or to multiple nodes.
When a node $v$ sends a message to all nodes, we say that $v$ \emph{broadcasts} this message.
A faulty node may fail to send messages as specified in the protocol.
The goal of \emph{reliable broadcast} is to ensure, under some conditions, that all honest nodes eventually \emph{deliver} a certain message, i.e., mark it as the accepted outcome of the reliable broadcast. The required properties of a reliable broadcast protocol are stated formally in the following definition.

\begin{definition}[Reliable broadcast]
\label{def:reliable_broadcast}
A \emph{reliable broadcast protocol} is a distributed protocol to send a message $m$ from a specific node called the \emph{sender} to all nodes with the following properties.
\begin{itemize}
\item \textbf{Validity}: If the sender is honest and broadcasts $m$, then every honest node eventually delivers $m$.
\item \textbf{Agreement}: If two honest node deliver messages $m$ and $m'$, then $m = m'$.
\item \textbf{Integrity}: Every honest node delivers at most one message $m$.
\item \textbf{Totality}: If an honest node delivers $m$, then all honest nodes eventually deliver $m$.
\end{itemize}
\end{definition}

In general, nodes may engage in multiple reliable broadcasts in parallel. In this case, the integrity condition effectively requires that messages are attributable to a uniquely identifiable execution context. In the following, we assume that every message implicitly contains an identifier and that each node runs a separate instance of the protocol for each identifier.
The sender for a particular execution context is assumed to be globally known.

As mentioned above, the objective is to reliably broadcast some message $m$, where the notation $|m|$ is used to denote the size of the message in bits.
Instead of sending $m$, the algorithms introduced in \S\ref{sec:algorithm} and \S\ref{sec:variant} send \emph{fragments} of $m$, which can be either chunks of $m$ itself or encoded data derived from $m$.
Without loss of generality, we assume that there is some upper bound $\ell_{max}$ on the message size, which in turn limits the fragment size.
State-of-the-art algorithms use either \emph{erasure codes}, which can handle missing data, or \emph{error-correcting codes}, which support the correction of erroneous data.
The algorithms in this paper exclusively use an $(n, k)$-erasure codes with optimal reception efficiency, i.e., exactly $k$ out of $n$ symbols are sufficient to reconstruct a message of $k$ symbols.
Throughout this paper, the parameters $n$ and $k$ correspond to the total number of nodes and the number of honest nodes, i.e., $k \coloneqq 2t+1$. We assume that each node has access to the routine \texttt{get\_fragments} that takes a message $m$ as input and returns a list of $n$ fragments of size $\frac{|m|}{k} = \frac{|m|}{2t+1} < \frac{3}{2}\frac{|m|}{n}$ each. Furthermore, the nodes use the routine \texttt{recover\_message} to recover $m$ given any subset $F$ of the $n$ fragments of size $|F| \ge k = 2t+1$. This routine is assumed to always return \emph{some} message, which may simply be a bit string of zeroes in case of an error, e.g., when the input contains fragments of different sizes.

In addition to fragments, the nodes also send hashes, identifiers, and, in the case of the algorithm presented in \S\ref{sec:variant}, also signatures and signature shares.
The hashes are assumed to be \emph{cryptographically strong} in the sense that it is infeasible to find hash collisions (except with negligible probability). Naturally, it is also assumed that it is computationally infeasible to spoof signatures or signature shares. Since the cryptographic properties only hold for hashes, signatures, and shares of a certain minimum size, we introduce a security parameter $\kappa$ and define that the size of all data types other than fragments is bounded by $O(\kappa)$.

Reliable broadcast algorithms are evaluated against multiple complexity measures.
For a specific number $n$ of nodes and message size $\ell$, the \emph{communication complexity} $\mathcal{C}(n,\ell)$ of an algorithm is the total number of bits sent by all honest nodes in the worst case assuming fewer than $n/3$ Byzantine nodes.
If the sender is Byzantine, we define that $\ell$ is the message size corresponding to the largest fragment sent by an honest node.
Given that at least $n\ell$ bits need to be transferred and we are concerned with the overhead for large messages, the primary goal is to minimize
$\mathcal{L}(n) \coloneqq \lim_{\ell\rightarrow \infty} \frac{\mathcal{C}(n,\ell)}{\ell n}$.
Additionally, the \emph{time complexity} of an algorithm is of great practical importance, measuring the duration of an execution when normalizing the maximum message delay to $1$ time unit. According to Definition~\ref{def:reliable_broadcast}, an execution with a Byzantine sender may never terminate. Therefore, we restrict our attention to the case where the sender is honest.
Due to the validity condition, this \emph{good-case} time complexity must be bounded.
Note that both presented algorithms ensure that the number of \emph{communication rounds} is also bounded given a Byzantine sender.
Lastly, the \emph{space complexity} is considered as well, which is defined as the number of bits that any honest node stores during the execution of the algorithm in the worst case.

\section{Algorithm}
\label{sec:algorithm}

\subsection{Overview}
The algorithm, referred to as $\mathcal{A}_{bit}$, effectively works as outlined in \S\ref{sec:introduction} with a few important additions.
In the first step, the sender disseminates the fragments. The validity of fragments is verified using Merkle proofs (as in \cite{cachin2005}). After validating the received fragment, each honest node broadcasts a \emph{proposal} to accept the message with the root hash in the received Merkle proof. The proposals serve to keep the communication complexity low and ensure that the integrity property holds.
Honest nodes broadcast their fragments only if they receive at least $2t+1$ proposals for the corresponding root hash. Lastly, if an honest node receives at least $2t+1$ fragments and manages to reconstruct the correct message, it first carries out the crucial step of sending the $j^{th}$ fragment to node $v_j$ if it did not receive any fragment from $v_j$ to guarantee that the totality property holds in case of a Byzantine sender before delivering the message.

\subsection{Description}

Algorithm $\mathcal{A}_{bit}$ makes use of the routines \texttt{get\_fragments} and \texttt{recover\_message} introduced in \S\ref{sec:model} to generate fragments of a given message and recover the message given at least $2t+1$ valid fragments, respectively. Furthermore, it requires routines for Merkle tree operations, specifically, \texttt{get\_merkle\_root} yields the Merkle root hash for a given set of fragments, \texttt{get\_merkle\_proof} returns the Merkle proof for a specific fragment, and \texttt{valid\_merkle\_proof} indicates whether a given Merkle proof is valid.

Each node executing $\mathcal{A}_{bit}$ locally maintains the data structures $\FragSet$, $\ReceiveSet$, $\HashSet$, and $\ProposeSet$.
Let $\mathcal{F}$, $\Pi$, and $\mathcal{H}$ denote the set of all possible fragments, Merkle proofs, and hashes, respectively.
The hash map $\FragSet: \mathcal{H} \times V \rightarrow \mathcal{F} \times \Pi$ stores the fragment $f_j \in \mathcal{F}$ and Merkle proof $\pi_j \in \Pi$ of node $v_j$ for the message with the root hash $h \in \mathcal{H}$. If $(f_j, \pi_i)$ is locally available, then $\FragSet(h, v_j) = (f_j, \pi_j)$, and $\FragSet(h, v_j) = \bot$ otherwise. This hash map is used to collect fragments with the goal of eventually recovering the corresponding message. We further define that $\FragSet(h) \coloneqq \cup_{v_j \in V} \FragSet(h, v_j)$ denotes the set of all collected $(v_j, \pi_j)$ pairs for hash $h$.
The hash map $\ReceiveSet: \mathcal{H} \rightarrow 2^V$ stores the nodes from which a fragment for root hash $h$ has been obtained. This hash map is needed to determine which nodes may still be missing their fragments.
Conversely, for any node $v$, the hash map $\HashSet: V \rightarrow 2^\mathcal{H}$ contains the set of hashes for which a proposal or fragment was received from node $v$.
The purpose of this data structure is to bound both the communication and space complexity.
Lastly, $\ProposeSet: \mathcal{H} \rightarrow 2^V$ indicates which nodes have proposed the delivery of the message associated with the Merkle root hash $h$.
The collected proposals are used to ensure that only fragments with sufficient support are broadcast.
All hash maps are initially empty. In addition to the hash maps, each node further uses the Boolean variable $done$, initially $false$, to capture the information whether the execution has terminated, either with or without the delivery of a message. In the latter case, the totality condition implies that no honest node will deliver a message for this execution.
Moreover, let $\mathbf{H} \coloneqq \cup_{v \in V} \HashSet(v)$ and $h_{max} \coloneqq \arg\max_{h \in \mathbf{H}} |\ProposeSet(h)|$, breaking ties arbitrarily if there is no single hash for which the most proposals have been received.

We distinguish between ``triggered actions'', which are executed when a function is invoked or a message is received, and ``state-based actions'', which occur when some conditions hold for the local state. This separation facilitates the description and analysis of the algorithm's properties.
The triggered actions are formally stated in Algorithm~\ref{algo:algorithm_bit_triggered_actions} and discussed next.

\begin{algorithm}[t]
\begin{algorithmic}
\caption{$\mathcal{A}_{bit}$:
Triggered actions at node $v_i$. Initially, $\FragSet = \ReceiveSet = \HashSet = \ProposeSet = \{\}$.}
\label{algo:algorithm_bit_triggered_actions}
\IF {\texttt{reliable\_broadcast($m$)} invoked \AND $v_i = \sender$}
\STATE ($f_1,\ldots, f_n$) $\coloneqq$ \texttt{get\_fragments}($m$)
\STATE $h$ $\coloneqq$ \texttt{get\_merkle\_root}($f_1,\ldots, f_n$)
\FOR {$v_j \in V$}
\STATE $\pi_j \coloneqq$ \texttt{get\_merkle\_proof}($(f_1, \ldots, f_n), j$)
\STATE \textbf{send} $fragment(h, j, f_j, \pi_j)$ to $v_j$
\ENDFOR
\ENDIF
\STATE
\IF {\textbf{received} $fragment(h, j, f_j, \pi_j)$ from $v_k$ \AND ($j=i$ \OR $j=k$)}
\IF {($|\HashSet(v_k)| < 2$ \OR $h \in \HashSet(v_k)$) \AND \texttt{valid\_merkle\_proof}($h, f_j, j, \pi_j$)}
\STATE $\HashSet(v_k) \coloneqq \HashSet(v_k) \cup \{h\}$, $\ReceiveSet(h) \coloneqq
\ReceiveSet(h) \cup \{v_k\}$, $\FragSet(h, v_j) \coloneqq (f_j, \pi_j)$
\IF {$i = j$ \AND first fragment from $v_k = \sender$}
\STATE \textbf{broadcast} $proposal(h)$
\ENDIF
\ENDIF
\ENDIF
\STATE
\IF {\textbf{received} $proposal(h)$ from $v_k$}
\IF {$|\HashSet(v_k)| < 2$ \OR $h \in \HashSet(v_k)$}
\STATE $\HashSet(v_k) \coloneqq \HashSet(v_k) \cup \{h\}$, $\ProposeSet(h) \coloneqq \ProposeSet(h) \cup \{v_k\}$
\ENDIF
\ENDIF
\end{algorithmic}
\end{algorithm}

Given a message $m$, a reliable broadcast is triggered by invoking the routine \texttt{reliable\_broadcast} at node $v_i = \sender$, which performs the standard steps of generating $n$ fragments $f_1,\ldots, f_n$, computing the Merkle root hash $h$, and then transmitting the \emph{fragment message} $fragment(h, j, f_j, \pi_j)$, which contains the corresponding Merkle proof $\pi_j$, to $v_j$ for all $j \in \{1,\ldots,n\}$.
A node $v_i$ only accepts a received message $fragment(h, j, f_j, \pi_j)$ from some node $v_k$ under the following conditions.
First, the sender $v_k$ must have sent the fragment of the recipient $v_i$ (i.e., $j = i$) or its own fragment (i.e., $j = k$). Second, the sender has not sent messages for two Merkle root hashes other than $h$ before (formally, $|\HashSet(v_k)| < 2$ or $h \in \HashSet(v_k)$) and, lastly, the Merkle proof $\pi_j$ in the received message is valid.
If all conditions are met, the hash maps $\HashSet$ and $\ReceiveSet$ are updated by adding $h$ and $v_k$, respectively, and the fragment and Merkle proof are stored ($\FragSet(h, v_j) \coloneqq (f_j, \pi_j)$).
In the final step, if the recipient received its fragment (i.e., $j = i$) and it is the first fragment received from the dedicated sender (i.e., $v_k = \sender$), then $v_i$ broadcasts $proposal(h)$.
Whenever a proposal of the form $proposal(h)$ is received from node $v_k$, it is again only accepted if $v_k$ did not send messages associated with two other hashes before, in which case $h$ is added to $\HashSet(v_k)$ and $v_k$ is added to $\ProposeSet(v_k)$.

Algorithm~\ref{algo:algorithm_bit_state_actions} describes the state-based actions of algorithm $\mathcal{A}_{bit}$. An honest node $v_i$ broadcasts its fragment $f_i$, at most once, after collecting $2t+1$ proposals for the corresponding Merkle root hash $h$.
By contrast, a node $v_i$ broadcasts a proposal for a specific root hash $h$ not only when it receives its fragment from the $\sender$ but also when at least $t+1$ fragments for the root hash $h = h_{max}$ have been received and $v_i$ has not broadcast the proposal before.
Lastly, if a node $v_i$ receives at least $2t+1$ fragments and proposals for $h_{max}$, it
recovers the message $m$ and, additionally, recomputes the fragments and the corresponding root hash. If the computed hash matches $h_{max}$, node $v_i$ concludes that the fragments are valid and the message $m$ can be delivered. In this case, $v_i$ first recomputes the Merkle proof $\pi_j$ and sends $fragment(h_{max},j,f_j,\pi_j)$ to each $v_j \in V \setminus \ReceiveSet(h_{max})$ before calling \texttt{deliver($m$)} to deliver the message.
Note that $done$ is set to $true$ after executing these steps even if $h \neq h_{max}$ because it is computationally infeasible to find any $2t+1$ fragments such that equality holds.

\begin{algorithm}[t]
\begin{algorithmic}
\caption{$\mathcal{A}_{bit}$: State-based actions at node $v_i$. Initially, $\FragSet = \ReceiveSet = \ProposeSet = \{\}$, $done = false$. Let $h_{max} \coloneqq \arg\max_{h \in \mathbf{H}} |\ProposeSet(h)|$.}
\label{algo:algorithm_bit_state_actions}
\IF {$|\ProposeSet(h_{max})| \ge 2t+1$ \AND \NOT broadcast $(f_i, \pi_i) \coloneqq \FragSet(h_{max},v_i) \ne \bot$ before}
\STATE \textbf{broadcast} $fragment(h_{max}, f_i, i, \pi_i)$
\ENDIF
\STATE
\IF {($|\FragSet(h_{max})| \ge t+1$ \AND \NOT broadcast $proposal(h_{max})$ before}
\STATE \textbf{broadcast} $proposal(h_{max})$
\ENDIF
\STATE
\IF {$|\ProposeSet(h_{max})| \ge 2t+1$ \AND $|\FragSet(h_{max})| \ge 2t+1$ \AND \NOT $done$}
\STATE $m \coloneqq$ \texttt{recover\_message($\FragSet(h_{max})$)}
\STATE ($f_1,\ldots, f_n$) $\coloneqq$ \texttt{get\_fragments}($m$)
\STATE $h \coloneqq$ \texttt{get\_merkle\_root($f_1,\ldots, f_n$)}
\IF {$h = h_{max}$}
\FOR {$v_j \in V \setminus \ReceiveSet(h_{max})$\rlap{\hspace{100pt}\smash{$\left.\begin{array}{@{}c}\\{}\\{}\\{}\\{}\\{}\\{}\\{}\\{}\end{array}\right\}\begin{tabular}{l}Execute delivery\end{tabular}$}}}
\STATE $\pi_j \coloneqq$ \texttt{get\_merkle\_proof}(($f_1,\ldots, f_n$), $j$)
\STATE \textbf{send} $fragment(h_{max}, j, f_j, \pi_j)$ to $v_j$
\ENDFOR
\STATE \texttt{deliver($m$)}
\ENDIF
\STATE $done \coloneqq true$
\ENDIF
\end{algorithmic}
\end{algorithm}

\subsection{Analysis}

In this section, we prove the correctness of $\mathcal{A}_{bit}$ as well as its communication, time, and space complexity.
A series of lemmas is used to simplify the proof structure.
The first lemma is concerned with the fragments that honest nodes broadcast, showing that the proposal mechanism ensures that honest nodes only broadcast fragments for one specific root hash.

\begin{lemma}
\label{lemma:one_fragment_hash}
If honest nodes broadcast fragments for root hashes $h$ and $h'$, then $h = h'$.
\end{lemma}
\begin{proof}
Without loss of generality, let $h$ be the first hash for which $2t+1$ proposals are received at some node $v$. Thus, there are at least $t+1$ honest nodes that received their fragments from the $\sender$ and then broadcast their first proposal for root hash $h$.

Let $v'$ be the first node that broadcasts a fragment for a hash $h' \ne h$. Since it must hold that $|\ProposeSet(h')| \ge 2t+1$ at node $v'$, there must be at least $t+1$ honest nodes that broadcast the proposal for root hash $h'$. Consequently, there must be an honest node $v^*$ that first proposed $h$ and then $h'$. However, $v^*$ only sends a proposal for $h'$ if $|\FragSet(h')| \ge t+1$ according to Algorithm~\ref{algo:algorithm_bit_state_actions}, which implies that an honest node must have sent a fragment for root hash $h'$ before, a contradiction to the assumption that $v'$ is the first such node.
\end{proof}

Since algorithm $\mathcal{A}_{bit}$ imposes restrictive rules for the acceptance of received messages, we must show that honest nodes always accept messages from other honest nodes regardless of the messages of Byzantine nodes.

\begin{lemma}
\label{lemma:accept_all}
If an honest node $v$ sends a proposal or fragment message to an honest node $v'$, then $v'$ accepts and processes the received message.
\end{lemma}
\begin{proof}
According to Algorithm~\ref{algo:algorithm_bit_triggered_actions}, $v'$ accepts messages associated with at most two different hashes for any sender $v$.
After an initial proposal for some root hash $h$, $v$ may send a second proposal for $h'$ if $|\FragSet(h')| \ge t+1$, i.e., there is at least one honest node that has broadcast a fragment for root hash $h'$. Assume that $v$ sends another proposal for a different root hash $h''$, which again implies that at least one honest node must have broadcast a fragment for this root hash, a contradiction to Lemma~\ref{lemma:one_fragment_hash}.

Regarding the transmission of fragments, Lemma~\ref{lemma:one_fragment_hash} also implies that an honest node only sends fragments for one root hash. Assume that $v$ sends fragments for a root hash $h''$ that differs from the hashes $h$ and $h'$ for which it sends proposals. In this case, $v$ must have received at least $t+1$ fragments for root hash $h$ or $h'$. Hence it follows that an honest node broadcast a fragment for a root hash other than $h''$, again contradicting Lemma~\ref{lemma:one_fragment_hash}.
\end{proof}

An important criterion for the totality property is that all honest eventually obtain sufficiently many fragments of a message $m$ and proposals for the corresponding root hash if there is an honest node that delivers $m$. The last lemma states that this is indeed the case.

\begin{lemma}
\label{lemma:honest_nodes_get_fragments}
If an honest node $v$ delivers $m$ with root hash $h$, all honest nodes will eventually store at least $2t+1$ proposals for $h$ and $2t+1$ fragments of $m$.
\end{lemma}
\begin{proof}
Assume that $v$ delivers $m$ with root hash $h$.
Since $|\FragSet(h)| \ge 2t+1$, $v$ received fragments of $m$ from at least $t+1$ honest nodes, i.e., every honest node eventually receives at least $t+1$ fragments of $m$. As a result, every honest node broadcasts $proposal(h)$ at some point and thus $|\ProposeSet(h)| \ge 2t+1$ eventually holds at all honest nodes.

Node $v$ adds all nodes from which it received fragments for root hash $h$
to the set $\ReceiveSet(h)$. For any node $v' \in \ReceiveSet(h)$ it holds that $v'$ either sent its own fragment or $v$'s fragment. However, the latter case also implies that $v'$ must have its own fragment.
According to Algorithm~\ref{algo:algorithm_bit_triggered_actions}, $v$ sends $f_j$ to $v_j$ for all $v_j \in V \setminus \ReceiveSet(h)$, guaranteeing that these nodes eventually receive and, due to Lemma~\ref{lemma:accept_all}, store their fragments as well. Thus, it eventually holds that
$|\ProposeSet(h)| \ge 2t+1$ and $\FragSet(h) \ne \bot$ at all honest nodes, causing them to broadcast their fragments. Consequently, every honest node will eventually receive at least $2t+1$ fragments of $m$.
\end{proof}

We are now in the position to prove the following main result.

\begin{theorem}
Algorithm $\mathcal{A}_{bit}$ implements reliable broadcast in the asynchronous communication model with $t < n/3$ Byzantine nodes.
\end{theorem}
\begin{proof}
The four conditions of reliable broadcast are proved separately.

\textbf{Validity}: If the sender $v$ of a message $m$ with root hash $h$ is honest, it sends the fragment $f_j$ to $v_j$ for all $v_j \in V$. Subsequently, every honest nodes broadcasts $proposal(h)$. Since all honest nodes eventually receive at least $2t+1$ proposals for hash $h$ (and only for hash $h$), they all broadcast their fragments. Thus, eventually $|\FragSet(h)| \ge 2t+1$ and $|\ProposeSet(h)| \ge 2t+1$ holds at all honest nodes, triggering the delivery of $m$.

\textbf{Agreement}: Assume for the sake of contradiction that honest nodes $v$ and $v'$ deliver distinct messages $m$ and $m'$ with root hashes $h$ and $h'$, respectively. Since $|\FragSet(h)| \ge 2t+1$ at $v$ and $|\FragSet(h')| \ge 2t+1$ at $v'$, there must be honest nodes that have broadcast fragments for $h$ and $h'$, which contradicts Lemma~\ref{lemma:one_fragment_hash}.

\textbf{Integrity}: Lemma~\ref{lemma:one_fragment_hash} implies that there can only be sufficiently many fragments for at most one hash and hence at most one message can be delivered.

\textbf{Totality}: Let $v$ be a node that delivers a message $m$ with root hash $h$. Due to Lemma~\ref{lemma:honest_nodes_get_fragments}, it holds eventually that $|\FragSet(h)| \ge 2t+1$ and $|\ProposeSet(h')| \ge 2t+1$ at all honest nodes. Since $v$ managed to verify the correctness of the fragments and the corresponding root hash, the verification also succeeds at all other nodes. Thus, all honest nodes eventually deliver $m$.
\end{proof}

The following theorem states the communication complexity of algorithm $\mathcal{A}_{bit}$.

\begin{theorem}[Communication complexity]
\label{theorem:bit_communication_complexity}
It holds for algorithm $\mathcal{A}_{bit}$ that $\mathcal{L}(n) = 2$.
\end{theorem}
\begin{proof}
Let $|f| \coloneqq \frac{3}{2}|m|/n + O(\kappa\log(n))$ denote the size of the largest fragment message sent by an honest node. If the sender is honest, $|f|$ corresponds to the size of a fragment message containing a valid fragment of message $m$.
The sender may send fragment messages to all other nodes for a total of at most $(n-1)|f|$ bits.
Each honest node may broadcast one fragment message containing its fragment and at most two proposals of size $O(\kappa)$.
After reconstructing the message $m$, an honest node further sends a fragment message containing $f_j$ to all $v_j \in V \setminus \ReceiveSet(h)$, where $|V \setminus \ReceiveSet(h)| \le t < n/3$. Not counting messages that nodes send to themselves, the communication complexity $\mathcal{C}(n, |m|)$ is therefore upper bounded by
\begin{displaymath}
(n - 1)|f| + n \cdot (n - 1 + t)|f| + O(\kappa n^2) < 2n|m| + O(\kappa n^2\log(n)),
\end{displaymath}
and therefore $\mathcal{L}(n) = \lim_{|m|\rightarrow \infty} \frac{1}{n|m|}\mathcal{C}(n, |m|) = 2$.
\end{proof}

As stated in \S\ref{sec:model}, we consider the good-case time complexity with an honest sender. The following theorem states the time complexity of $\mathcal{A}_{bit}$ for this case.

\begin{theorem}[Time complexity]
\label{theorem:bit_time_complexity}
If the sender is honest, $\mathcal{A}_{bit}$ has a time complexity of $3$.
\end{theorem}
\begin{proof}
If the sender is honest, all honest nodes receive their fragments and send the proposal for the corresponding root hash after at most $1$ time unit. Thus, after at most $2$ time units, all honest nodes get at least $2t+1$ proposals and broadcast their fragments. After at most $3$ time units, the honest nodes receive at least $2t+1$ fragments and deliver the message.
\end{proof}

If the sender is Byzantine, an honest node may still deliver $m$ eventually. It is easy to see that every honest node delivers $m$ at most $3$ time units later in this case.

As far as the space complexity is concerned, recall that $\ell_{max}$ denotes the largest permissible message size, which is a lower bound on the space complexity.
Given this bound, the following result shows that honest nodes require little additional storage space.

\begin{theorem}[Space complexity]
\label{theorem:bit_space_complexity}
Algorithm $\mathcal{A}_{bit}$ has a space complexity of $2\ell_{max} + O(n\kappa)$.
\end{theorem}
\begin{proof}
While honest nodes send fragments for at most one root hash, Byzantine nodes cannot send fragments for more than two hashes, resulting in a total of at most $t\cdot 2 + (2t+1)\cdot 1 < \frac{4}{3}n$ fragments of size $\frac{3}{2}\ell_{max}/n$ each. Thus, the space required to store received fragments is upper bounded by $2\ell_{max}$.
The other data structures merely contain identifiers and hashes of size $O(\kappa)$. Since they may contain entries for at most $2$ hashes per node, it follows that the total size of these data structures is bounded by $O(n\kappa)$.
\end{proof}

Note that the space complexity can trivially be lowered to $\frac{3}{2}\ell_{max} + O(n\kappa)$ by introducing the rule that received fragments for different hashes from the same node are rejected.
This rule is omitted from Algorithm~\ref{algo:algorithm_bit_triggered_actions} for ease of exposition.

\section{Variant with Threshold Signatures}
\label{sec:variant}

\subsection{Overview}

In this section, a variant of $\mathcal{A}_{bit}$, called $\mathcal{A}_{sig}$, is introduced with an optimal time complexity of $2$ while still guaranteeing that $\mathcal{L}(n) < 3$.
The challenge is that this time complexity cannot be reached when sending proposals to agree on a root hash first before sending fragments. However, if fragments are broadcast right away, a Byzantine sender may send fragments for different root hashes to different nodes. Naturally, honest nodes cannot simply broadcast any received fragments without an adverse effect on the communication complexity.

Algorithm $\mathcal{A}_{sig}$ addresses this challenge by using \emph{threshold signatures}, i.e., we assume that each node has the keying material to \emph{threshold-sign} messages and that every node can verify the validity of \emph{signature shares} and signatures of all nodes. At least $2t+1$ signature shares are required to create a valid signature.
The proposals of $\mathcal{A}_{bit}$ are replaced by messages containing a root hash and a signature share or a signature. The crucial advantage is that a signature serves as proof that the sender has received signature shares from $2t+1$ nodes, ensuring that the execution can make progress whenever an honest nodes meets the conditions to deliver $m$.

\subsection{Description}

Algorithm $\mathcal{A}_{sig}$ uses \emph{signed proposal messages} of the form $proposal(h, \sigma)$, where $\sigma$ is either a \emph{signature share} of $h$ of a particular node or a \emph{signature} of $h$ derived from at least $2t+1$ signature shares. Note that, in practice, it is a signature of $h$ and the identifier in order to tie the signature to a particular execution.
The algorithm still uses the hash maps $\FragSet$, $\HashSet$, and $\ReceiveSet$. Signature shares are collected in the hash map $\SigSet: \mathcal{H} \times V \rightarrow \mathcal{S}$, where $\mathcal{S}$ denotes the set of all possible signature shares.
We further define that $S(h) \coloneqq \cup_{v \in V} \SigSet(h,v)$.
In this section, $h_{max}$ is defined as $h_{max} \coloneqq \arg\max_{h \in \mathbf{H}} |\SigSet(h)|$.
In addition to $done$, each node further stores $h^*$, the root hash of the message that will be delivered if set, and the corresponding signature $\sigma^*$ (initially, $h^* = \sigma^* = \bot$).

\begin{algorithm}[t]
\begin{algorithmic}
\caption{$\mathcal{A}_{sig}$: Triggered actions at node $v_i$. Initially, $\FragSet = \ReceiveSet = \HashSet = \SigSet = \{\}$, $h^* = \sigma^* = \bot$.}
\label{algo:algorithm_sig_triggered_actions}
\IF {\texttt{reliable\_broadcast($m$)} invoked \AND $v_i = \sender$}
\STATE Execute \texttt{reliable\_broadcast($m$)} of Algorithm~\ref{algo:algorithm_bit_triggered_actions}
\ENDIF
\STATE
\IF {\textbf{received} $fragment(h, j, f_j, \pi_j)$ from $v_k$ \AND ($j=i$ \OR $j=k$)}
\IF {($|\HashSet(v_k)| < 2$ \OR $h \in \HashSet(v_k)$) \AND \texttt{valid\_merkle\_proof}($h, f_j, j, \pi_j$)}
\STATE $\HashSet(v_k) \coloneqq \HashSet(v_k) \cup \{h\}$, $\ReceiveSet(h) \coloneqq \ReceiveSet(h) \cup \{v_k\}$, $\FragSet(h, v_j) \coloneqq (f_j, \pi_j)$
\IF {$i = j$ \AND first fragment from $v_k = \sender$}
\STATE $\sigma_i \coloneqq$ \texttt{threshold\_sign($h$)}
\STATE \textbf{broadcast} $[proposal(h, \sigma_i)$, $fragment(h, f_i, i, \pi_i)]$
\ENDIF
\ENDIF
\ENDIF
\STATE
\IF {\textbf{received} $proposal(h, \sigma)$ from $v_k$}
\IF {\texttt{valid\_signature($h, \sigma$)}}
\STATE $\sigma^* \coloneqq \sigma$, $h^* \coloneqq h$
\ELSIF {(($|\HashSet(v_k)| < 2$ \OR $h \in \HashSet(v_k)$) \AND \texttt{valid\_share}($h, v_k, \sigma$)}
\STATE $\HashSet(v_k) \coloneqq \HashSet(v_k) \cup \{h\}$, $\SigSet(h, v_k) \coloneqq \sigma$
\ENDIF
\ENDIF
\end{algorithmic}
\end{algorithm}

The steps of algorithm $\mathcal{A}_{sig}$ for triggered actions are shown in Algorithm~\ref{algo:algorithm_sig_triggered_actions}.
When \texttt{reliable\_broadcast} is invoked, exactly the same steps as in Algorithm~\ref{algo:algorithm_bit_triggered_actions} are executed.
A fragment message is handled the same way as in $\mathcal{A}_{bit}$ apart from two differences: $h$ is signed and both $proposal(h, \sigma_i)$ and $fragment(h, f_i, i, \pi_i)$ are broadcast if it is the first (valid) fragment received from the sender.
When receiving $proposal(h, \sigma)$ from some node $v_k$, the action depends on $\sigma$. If it is a valid signature, $\sigma^*$ and $h^*$ are set to $\sigma$ and $h$, respectively. Otherwise, if it is a valid signature share and $\HashSet(v_k)$ does not contain two hashes different from $h$, then $h$ is added to $\HashSet(v_k)$ and $\SigSet(h, v_k)$ is set to $\sigma$.

The state-based actions are shown in Algorithm~\ref{algo:algorithm_sig_state_actions}.
If at least $2t+1$ signature shares have been collected for some root hash $h_{max}$, then $\sigma^*$ is set to the computed signature and $h^*$ is set to $h_{max}$. Node $v_i$ broadcasts $f_i$ (including $\pi_i$) if $h^*$ is set, $f_i$ is locally available, and it has not been broadcast before.
Lastly, if $h^* \neq \bot$ and $|\FragSet(h^*)| \ge 2t+1$, $v_i$ broadcasts $proposal(h^*, \sigma^*)$ before executing the same delivery steps as in Algorithm~\ref{algo:algorithm_bit_state_actions}.

\subsection{Analysis}

It is evident from the description of algorithm $\mathcal{A}_{sig}$ that honest nodes must set $h^*$ to the same hash as otherwise the integrity property may be violated.
The following lemma states that honest nodes indeed cannot set $h^*$ to different hashes.

\begin{lemma}
\label{lemma:sig_same_h_id}
If honest nodes set $h^*$, they set it to the same hash.
\end{lemma}
\begin{proof}
According to Algorithm~\ref{algo:algorithm_sig_triggered_actions}, only the first message from the $\sender$ is threshold-signed, i.e., every honest node provides a signature share for at most one hash. However, setting $h^*$ to some hash $h$ requires $2t+1$  signature shares of $h$. If we assume that there are two different hashes for which $2t+1$ signature shares have been collected, there must be at least one honest node that threshold-signed two different hashes, which contradicts the rule that honest nodes only threshold-sign at most one hash.
\end{proof}

\begin{algorithm}[t]
\begin{algorithmic}
\caption{$\mathcal{A}_{sig}$: State-based actions at node $v_i$. Initially, $\FragSet = \ReceiveSet = \SigSet = \{\}$, $h^* = \sigma^* = \bot$, and $done = false$. Let $h_{max} \coloneqq \arg\max_{h \in \mathcal{H}} |\SigSet(h)|$.}
\label{algo:algorithm_sig_state_actions}
\IF {$|\SigSet(h_{max})| \ge 2t+1$ \AND $h^* = \bot$}
\STATE $\sigma^* \coloneqq$ \texttt{compute\_signature($\SigSet(h_{max})$)}
\STATE $h^* \coloneqq h_{max}$
\ENDIF
\STATE
\IF {$h^* \ne \bot$ \AND \NOT broadcast $(f_i, \pi_i) \coloneqq \FragSet(h^*,i) \ne \bot$ before}
\STATE \textbf{broadcast} $fragment(h^*, f_i, i, \pi_i)$
\ENDIF
\STATE
\IF {$h^* \ne \bot$ \AND $|\FragSet(h^*)| \ge 2t+1$ \AND \NOT $done$}
\STATE \textbf{broadcast} $proposal(h^*,\sigma^*)$
\STATE Execute delivery as in Algorithm~\ref{algo:algorithm_bit_state_actions}
\ENDIF
\end{algorithmic}
\end{algorithm}

A Byzantine sender may send fragments for different root hashes to the honest nodes, causing wasteful transmissions of fragments. However, algorithm $\mathcal{A}_{sig}$ ensures that the number of transmissions is bounded as the next lemma shows.

\begin{lemma}
\label{lemma:sig_two_hashes}
Every honest node sends fragments for at most two different root hashes.
\end{lemma}
\begin{proof}
An honest node $v_i$ broadcasts its fragment $f_i$ for some root hash $h$ when receiving $f_i$ from the $\sender$. Apart from this transmission, according to Algorithm~\ref{algo:algorithm_sig_triggered_actions}, $v_i$ only sends fragments associated with $h^*$. Since $h^*$ never changes once set, the claim follows.
\end{proof}

Since algorithm $\mathcal{A}_{sig}$ enforces similar restrictions for the acceptance of messages as algorithm $\mathcal{A}_{bit}$, Lemma~\ref{lemma:accept_all} holds for $\mathcal{A}_{sig}$ as well.

\begin{lemma}
\label{lemma:sig_accept_all}
Lemma~\ref{lemma:accept_all} (``If an honest node $v$ sends a proposal or fragment message to an honest node $v'$, then $v'$ accepts and processes the received message.'') holds for $\mathcal{A}_{sig}$.
\end{lemma}
\begin{proof}
Algorithm $\mathcal{A}_{sig}$ and $\mathcal{A}_{bit}$ share the property that messages for two different hashes are accepted from any node.
According to Algorithm~\ref{algo:algorithm_sig_triggered_actions}, an honest node broadcasts a signature and fragment message for the same hash $h$ when receiving its fragment from the $\sender$. Algorithm~\ref{algo:algorithm_sig_state_actions} states that any further transmission of a signature or fragment message must be for hash $h^* \ne \bot$. Since $h^*$ never changes when set, the claim follows.
\end{proof}

The following variant of Lemma~\ref{lemma:honest_nodes_get_fragments} holds for algorithm $\mathcal{A}_{sig}$.

\begin{lemma}
\label{lemma:sig_honest_nodes_get_fragments}
If an honest node delivers $m$ with root hash $h$, all honest nodes will eventually store at least $2t+1$ fragments of $m$ and set $h^*$ to $h$.
\end{lemma}
\begin{proof}
The same argument as for algorithm $\mathcal{A}_{bit}$ applies, which proves that every honest node gets its fragment eventually after an honest node $v$ has delivered a message $m$ for some hash $h$, which it accepts due to Lemma~\ref{lemma:sig_accept_all}.
Since $v$ must have computed and broadcast a valid signature $\sigma$, it eventually holds at all honest nodes that $h^* = h$ (and $\sigma^* = \sigma$). According to Lemma~\ref{lemma:sig_same_h_id}, honest nodes all set $h^*$ to the same value. Consequently, every honest node will eventually broadcast its fragment and thus also receive at least $2t+1$ fragments of $m$.
\end{proof}

As in \S\ref{sec:algorithm}, the first main result is that $\mathcal{A}_{sig}$ is a correct reliable broadcast algorithm.

\begin{theorem}
Algorithm $\mathcal{A}_{sig}$ implements reliable broadcast in the asynchronous communication model with $t < n/3$ Byzantine nodes.
\end{theorem}
\begin{proof}
The four conditions of reliable broadcast are again proved separately.

\textbf{Validity}: If the sender $v$ of a message $m$ with root hash $h$ is honest, it sends $f_j$ to $v_j$ for all $v_j \in V$. Every honest node $v_j$ broadcasts its fragment and signature share and, consequently, receives at least $2t+1$ signature shares and fragments eventually. It follows that every honest node eventually computes $\sigma^*$, sets $h^* \coloneqq h$, and proceeds to deliver $m$.

\textbf{Agreement}: If two honest nodes $v$ and $v'$ deliver different messages $m$ and $m'$, they must have set $h^*$ to different hashes, a contradiction to Lemma~\ref{lemma:sig_same_h_id}.

\textbf{Integrity}: Due to Lemma~\ref{lemma:sig_same_h_id}, every honest node delivers at most one message.

\textbf{Totality}: If an honest node $v$ delivers a message $m$ with root hash $h$, it eventually holds at all honest nodes that $h^* = h$ and $|\FragSet(h^*)| \ge 2t+1$ due to Lemma~\ref{lemma:sig_honest_nodes_get_fragments}. Since $v$ managed to verify the correctness of the fragments and the root hash, it holds that the verification must succeed at other nodes as well, which entails that every honest node delivers $m$.
\end{proof}

Compared to $\mathcal{A}_{bit}$, the communication complexity of algorithm $\mathcal{A}_{sig}$ is worse, which appears to be inevitable given that the nodes must broadcast fragments without delay to achieve an optimal time complexity.
However, the following lemma states that even a Byzantine sender cannot induce many (wasteful) transmissions of fragments.

\begin{lemma}
\label{lemma:sig_at_most_t}
At most $t$ honest nodes send fragments for $2$ different root hashes.
\end{lemma}
\begin{proof}
Let $S$ denote the set of honest nodes that send fragments for more than one root hash. Algorithm~\ref{algo:algorithm_sig_triggered_actions} dictates that honest nodes only send another fragment when $h^* \ne \bot$. If there are $t' \le t$ Byzantine nodes in the execution, at least $2t+1-t'$ honest nodes must receive their fragments for hash $h^*$ in order to obtain a signature $\sigma^*$ for this hash. Since none of these nodes sends fragments for other root hashes, we get that $|S| \le 3t+1 - (2t + 1 - t') - t' = t$.
\end{proof}

As the following theorem shows, $\mathcal{A}_{sig}$ achieves an overhead factor of $5/2$ for $|m|\rightarrow \infty$.

\begin{theorem}[Communication complexity]
\label{theorem:sig_communication_complexity}
It holds for algorithm $\mathcal{A}_{sig}$ that $\mathcal{L}(n)  = \frac{5}{2}$.
\end{theorem}
\begin{proof}
The dissemination of fragments by the sender requires at most $(n-1)|f|$ bits, where $|f|$ again denotes the size of the largest fragment message sent by an honest node.
Due to Lemma~\ref{lemma:sig_at_most_t}, at most $t$ honest nodes broadcast fragments twice, and the other honest nodes broadcast at most one fragment. Additionally, each node sends $O(\kappa n)$ bits for the signed proposal messages and $f_j$ to all $v_j \in V \setminus \ReceiveSet(h)$, where $|V \setminus \ReceiveSet(h)| \le t < n/3$ as before. Hence it follows that $\mathcal{C}(n, |m|)$ is at most
\begin{eqnarray*}
& (n-1)|f| + (2t+1)(n-1)|f| + 2t(n-1)|f| + t(n-1)|f| + O(\kappa n^2)\\
&< (n-1)|f| + \frac{5}{3}n(n-1)|f| + O(\kappa n^2)
< \frac{5}{2}n|m| + O(\kappa n^2\log(n)),
\end{eqnarray*}
where the first inequality holds because $(5t+1) < \frac{5}{3}(3t+1) = \frac{5}{3}n$. Thus, $\mathcal{L}(n) = \frac{5}{2}$.
\end{proof}

It is worth noting that the communication complexity is only worse compared to $\mathcal{A}_{bit}$ if the sender deliberately transmits conflicting fragments, i.e.,
$\mathcal{L}(n) = 2$ still holds otherwise.
The main advantage of $\mathcal{A}_{sig}$ over $\mathcal{A}_{bit}$ is its superior time complexity.

\begin{theorem}[Time complexity]
\label{theorem:sig_time_complexity}
If the sender is honest, $\mathcal{A}_{sig}$ has a time complexity of $2$.
\end{theorem}
\begin{proof}
If the sender is honest, all honest nodes receive their fragments within $1$ time unit and broadcast their signature shares and fragments. After $2$ time units, all honest nodes must have received at least $2t+1$ signature shares and fragments for the same root hash. Thus, all honest nodes compute the signature $\sigma^*$ and set $h^*$, triggering the delivery of $m$.
\end{proof}

If an honest node $v$ delivers $m$ despite a Byzantine sender, it is straightforward to show that each honest node delivers $m$ at most $2$ time units later.
While the time complexity of $\mathcal{A}_{sig}$ is lower, it has a slightly higher space complexity.

\begin{theorem}[Space complexity]
Algorithm $\mathcal{A}_{sig}$ has a space complexity of $\frac{5}{2}\ell_{max} + O(n\kappa)$.
\end{theorem}
\begin{proof}
According to Lemma~\ref{lemma:sig_at_most_t}, there are at most $t$ honest nodes that send their fragments for $2$ root hashes, i.e., there are at least $t+1$ honest nodes that only send their fragment once. Thus, an honest node may store $2t \cdot 2 + (t+1) \cdot 1 \le \frac{5}{3}n$ fragments of size $\frac{3}{2}\ell_{max}/n$ for a total of $\frac{5}{2}\ell_{max}$.
Algorithm $\mathcal{A}_{sig}$ also utilizes $\ReceiveSet$, containing at most $2$ node identifiers of size $O(\kappa)$ per node, and $\SigSet$, storing at most $2$ signature shares of size $O(\kappa)$ per node. Thus, the space complexity of these data structures is upper bounded by $O(n\kappa)$.
\end{proof}

\section{Practical Considerations}
\label{sec:optimizations}

In this section, optimizations that may be relevant for practical applications are discussed.

\noindent\textbf{Partially synchronous communication.}
Byzantine behavior and unpredictable message latencies are often exceptional situations in practice.
While $\mathcal{A}_{sig}$ and $\mathcal{A}_{bit}$ are both tailored to the asynchronous communication model, they can easily be adapted to the \emph{partially synchronous model} where periods of synchrony are assumed.
Given an upper bound $d$ on the message delay that holds for some periods of time, both algorithms can be modified to improve the communication complexity during these times in the absence of faults. Specifically, the additional constraint can be introduced that at least $\delta$ time must have passed since the first fragment was received before delivering $m$.
The following theorem states the effect of this modification on the communication complexity for an algorithm-specific $\delta$.

\begin{theorem}
If communication is synchronous from the start of the execution for $\delta=3d$
($\delta=2d$) time and there are no faults, then $\mathcal{L}(n) = \frac{3}{2}$ for the adapted version of $\mathcal{A}_{bit}$ ($\mathcal{A}_{sig}$).
\end{theorem}
\begin{proof}
As there are no faults by assumption, Theorem~\ref{theorem:bit_time_complexity} and Theorem~\ref{theorem:sig_time_complexity} imply that all nodes receive $n$ fragments after at most $3d$ ($2d$) time for $\mathcal{A}_{bit}$ ($\mathcal{A}_{sig}$).
Thus, $\ReceiveSet(h) = V$ for the root hash $h$ of $m$, which entails that the nodes do not  disseminate any additional fragments before delivering, resulting in a communication complexity of $(n-1)|f| + n(n-1)|f| + O(\kappa n) < n^2|f| + O(\kappa n) \le \frac{3}{2}|m|n + O(\kappa n^2\log(n))$ and thus $\mathcal{L}(n) = \frac{3}{2}$.
\end{proof}

Naturally, this modification does not improve any worst-case bounds because $\mathcal{L}(n) = 2$ still holds even for $t$ crash failures but it may be beneficial in practice nonetheless.

\noindent\textbf{Simplified structure.}
The second optimization only concerns $\mathcal{A}_{sig}$. Instead of handling fragments and proposals separately, signature shares and signatures can be appended to the fragment message. As a result, nodes only send and process messages of a single type.
Fragments and proposals are kept separate in algorithm $\mathcal{A}_{sig}$ to retain as much similarity to $\mathcal{A}_{bit}$ as possible in order to emphasize the key differences and simplify the presentation.
This modification requires the addition of a simple rule: a node appends its signature share as long as $h^* = \bot$ and the signature $\sigma^*$ otherwise. Note that $\sigma^* \ne \bot$ always holds when $h^* \ne \bot$. Signature shares and signatures are still processed as shown in Algorithm~\ref{algo:algorithm_sig_triggered_actions}.

According to Algorithm~\ref{algo:algorithm_sig_state_actions}, the signature is broadcast before executing the message delivery. This step can be omitted entirely when adding the rule above without violating any of the reliable broadcast conditions. Since the signature is broadcast to ensure totality, we briefly argue why this condition continues to hold. Let $v$ be an honest node that has delivered message $m$ with root hash $h$. If all honest nodes broadcast fragments for root hash $h$ at some point, then all honest nodes will eventually deliver $m$ even if they never receive a signature. On the other hand, if there is at least one honest node $v'$ that never sends a fragment for the right hash, then $v' \notin \ReceiveSet(h)$ at node $v$, and thus $v$ must send $v'$'s fragment and the signature to $v'$. Upon receipt of this message, $v'$ will broadcast its fragment and the signature, ensuring that all honest nodes eventually receive the signature, i.e., the signature is broadcast anyway in this case.

\label{sec:lower_bound}

The preceding section showed that an overhead factor of $\frac{3}{2}$ for large messages can  be attained in practice under normal network conditions and in the absence of faults.
In this section, we shed light on the question whether we can hope to further improve upon this bound. The question is answered in the negative, at least for a specific but relevant class of algorithms.

The lower bound holds in the \emph{synchronous communication model} where all nodes
perform some computation, send messages, and receive the sent messages within the same \emph{round} of a bounded and known duration. Furthermore, the lower bounds holds in the
\emph{crash failure model} where a faulty node stops executing at any point during the execution but it never deviates from the correct protocol execution until it fails.
In order to utilize the available bandwidth efficiently, it is beneficial to minimize the maximum bandwidth consumption over all nodes. If each node sends the same amount of data up to a constant factor, the algorithm is called \emph{balanced}~\cite{alhaddad2022}. Any imbalance beyond a constant factor is usually due to the sender transmitting more data, typically in the first round.
We focus on a broader class of ``weakly balanced'' algorithms where the sender can send arbitrarily sized messages to other nodes but at most $o(|m|n)$ bits overall in the first round. Note that each node sends $O(|m| + \kappa n\log(n))$ bits in both $\mathcal{A}_{bit}$ and $\mathcal{A}_{sig}$.

We further restrict our attention to reliable broadcast algorithms that, given an honest sender, guarantee that all honest nodes deliver $m$ after at most $3$ (synchronous) rounds.
This class of algorithms is interesting because it covers $\mathcal{A}_{bit}$, $\mathcal{A}_{sig}$, as well as other algorithms in the literature including the algorithm by Cachin and Tessaro.

Formally, a \emph{$(b,r)$-reliable broadcast algorithm} is defined as a reliable broadcast algorithm that sends messages of size at most $b$ bits in the first round and delivers the message of an honest sender at all honest nodes within at most $r$ rounds.

Let $b_j$ denote the bits (of entropy) that the sender sends to $v_j$ in the first round according to the given algorithm in some execution \emph{without any failures}. As mentioned before, we assume that\begin{equation}\label{eq:first_round}
\sum_{j=1}^n b_j \in o(|m|n).
\end{equation}

Let $b_{ij}$ further denote the bits that $v_i$ sends to $v_j$ in rounds $2,\ldots,r$ in the same execution. Due to the validity condition and the fact that we consider algorithms with a good-case time complexity of at most $r$, all honest nodes must receive $m$ by the end of round $r$.
Let $S_j$ denote the set $\{ v_i \in V \;|\; b_{ij} > 0\}$, i.e., the set of nodes that send some bits to $v_j$. Moreover, let $R \coloneqq \{ v_j \in V \;|\; b_j \ge |m|\}$ denote the set of nodes that receive at least $|m|$ bits in round $1$. Since $\sum_{j=1}^n b_j\in o(|m|n)$, it must hold that $|R| \in o(n)$.
Let $\bar{b}_j \coloneqq \frac{1}{|S_j|}\sum_{i=1}^n b_{ij}$ denote the average number of bits sent to $v_j$ from the nodes in $S_j$ in rounds $2,\ldots, r$.

In the following, we consider the case $r=3$, which permits a simple strategy to derive a lower bound because a node failing to receive expected messages in the second round can only inform other nodes about the failure in the third round. Requiring a fourth round to send the missing bits violates the requirement that all honest nodes deliver the message by round $3$ in case of a non-faulty sender. Therefore, the strategy is simply to mark those nodes as faulty that send many bits in rounds $2$ and $3$.
Given this strategy, it is easy to see that $|S_j| > t$ and $\bar{b}_j > 0$ must hold for any $v_j \in V \setminus R$ as otherwise $v_j$ may not receive $m$ by the end of round $3$.
The following lemma provides a lower bound on $\bar{b}_j$ for the case $|S_j| > t$.

\begin{lemma}
For all $j \in \{1,\ldots, n\}$, $|S_j| > t$, it must hold that
\begin{equation}
\label{eq:average_entropy}
\bar{b}_j \ge (|m| - b_j) / (|S_j| - t).
\end{equation}
\end{lemma}
\begin{proof}
Note that $b_j \ge |m|$ implies that $\bar{b}_j \ge 0$, which holds trivially.
Thus, we consider the case $b_j < |m|$. Assume that $\bar{b}_j < (|m| - b_j) / (|S_j| - t)$ and that the $t$ nodes in $S_j$ that send the largest number of bits are faulty and never send anything. Let $S_j'$ denote the remaining nodes that send bits in rounds $2$ and $3$. In this case, node $v_j$ receives
\begin{align*}
b_j + \sum_{v_i \in S_j'} b_{ij} &\le b_j + (|S_j| - t)\bar{b}_j < b_j + (|S_j| - t)(|m| - b_j) / (|S_j| - t) = |m|
\end{align*}
bits by the end of round $3$, which implies that $v_j$ does not receive the whole message.
\end{proof}

The following theorem states the main result that the overhead factor must be at least $\frac{3}{2}$ for the considered class of algorithms as $n$ tends to infinity.

\begin{theorem}
If communication is synchronous and there are $t < n/3$ faulty nodes, it holds that $\mathcal{L}(n, t) \ge \frac{3}{2}$
for every $(o(|m|n), 3)$-reliable broadcast algorithm as $n \rightarrow \infty$.
\end{theorem}
\begin{proof}
The communication complexity is at least
\begin{align*}
&\sum_{j=1}^n b_j + \sum_{j=1}^n \sum_{i=1}^n b_{ij} \ge \sum_{v_j \in V\setminus R} \sum_{i=1}^n b_{ij} = \sum_{v_j \in V\setminus R} |S_j| \bar{b}_j
\stackrel{\eqref{eq:average_entropy}}{\ge} \sum_{v_j \in V\setminus R} |S_j|\frac{|m| - b_j}{|S_j| - t}\\
&\ge \sum_{v_j \in V\setminus R} \frac{n}{n-t}(|m|-b_j)
= (n-|R|)\frac{n}{n-t}|m| - \sum_{v_j \in V\setminus R} \frac{n}{n-t} b_j\\
&\stackrel{\eqref{eq:first_round}}{\ge} (n-|R|)\frac{n}{n-t}|m| - o(|m|n)
\stackrel{n \rightarrow \infty}{=} n\frac{n}{n-n/3}|m| = \frac{3}{2}|m|n.  \qedhere
\end{align*}
\end{proof}

\section{Related Work}
\label{sec:related_work}

As mentioned in \S\ref{sec:introduction}, the first reliable broadcast algorithm, achieving a communication complexity of $O(|m|n^2)$, was presented by Bracha~\cite{bracha1987}.
An algorithm with a much improved bound of $O(|m|n + \kappa n^2\log(n))$, making use of erasure codes, was published 18 years later by Cachin and Tessaro~\cite{cachin2005}, making use of collision-resistant hash functions of size $\kappa$.
By contrast, Bracha's algorithm is  \emph{error-free}, i.e., it does not rely on any cryptographic assumptions.

\enlargethispage{\baselineskip}
Subsequently, error-free reliable broadcast algorithms with lower communication complexity were proposed, guaranteeing bounds of
$O(n|m| + n^4\log(n))$~\cite{patra2011} and $O(n|m| + n^3\log(n))$~\cite{nayak2020}. While the latter algorithm achieves a lower asymptotic communication complexity, it is not \emph{balanced} in that a single node, the sender, must transmit more bits than the other nodes. A reliable broadcast algorithm is said to be \emph{balanced} if all nodes send the same number of bits up to a constant factor~\cite{alhaddad2022}.
Note that the algorithms presented in \S\ref{sec:algorithm} and \S\ref{sec:variant} are both balanced.
The best known bound for error-free algorithms is
$O(n|m| + n^2\log(n))$~\cite{alhaddad2022}.
Abraham and Asharov proposed a probabilistic algorithm with a similar bound of $O(n|m| + n^2\log(n^3/\varepsilon))$, guaranteeing validity but the agreement and totality properties only hold with probability $1-\varepsilon$~\cite{abraham2022}.

It has also been shown how to achieve a communication complexity of $O(n|m| + \kappa n^2)$ using only collision-resistant hash functions~\cite{das2021}, improving upon the algorithm by Cachin and Tessaro. The downside of this algorithm is that it is not balanced and it has a higher computational cost. Algorithm $\mathcal{A}_{bit}$ introduced in \S\ref{sec:algorithm} falls into the category of reliable broadcast algorithms that rely on collision-resistant hash functions as well.
There are reliable broadcast algorithms that require other cryptographic primitives. For example, an algorithm has been proposed that achieves a communication complexity of $O(n|m| + \kappa n^2)$ but requires a trusted setup for a public key infrastructure and cryptographic accumulators~\cite{nayak2020}. This bound has been improved to $O(n|m| + \kappa n + n^2)$ using threshold signatures~\cite{alhaddad2022}. The algorithm $\mathcal{A}_{sig}$ defined in \S\ref{sec:variant} uses threshold signatures to achieve an optimal (good-case) time complexity.

The best known upper bounds, with and without cryptographic assumptions, are almost asymptotically tight considering the lower bound of $\Omega(n|m| + n^2)$~\cite{dolev1985}.
The algorithms in this paper improve the \emph{effective} communication complexity for large messages by reducing the constant of the first term.
Moreover, most error-free algorithms, excluding Bracha's algorithm, and also the algorithms that use more cryptographic tooling than hash functions, have a higher time complexity than $\mathcal{A}_{bit}$ and $\mathcal{A}_{sig}$, the latter having an optimal
time complexity~\cite{abraham2021b}.

Byzantine reliable broadcast has further been studied in a variety of models that differ from the model used in most related work. There is a probabilistic algorithm based on stochastic samples that allows each property to be violated with a fixed and small probability~\cite{guerraoui2019}.
An algorithm for a model with dynamic membership has also been proposed~\cite{guerraoui2020}. There is further work on \emph{consistent broadcast}, a variant of reliable broadcast without the totality property, and its applications~\cite{cachin2001, reiter1994}.
Lastly, in an effort to minimize the actual latency and bandwidth consumption, simulations have been used to show the efficacy when combining Bracha's algorithm with Dolev's reliable communication protocol~\cite{dolev1981} to disseminate messages reliably in partially connected networks~\cite{bonomi2021}.

\section{Conclusion}
\label{sec:conclusion}

The mechanism introduced in this paper lowers the communication complexity of Byzantine reliable broadcast for large messages. The presented algorithms, which utilize this mechanism, further guarantee near-optimal and optimal bounds on the time complexity while also keeping the space complexity close to optimal.
As mentioned in \S\ref{sec:introduction}, numerous applications make use of reliable broadcast as a subroutine. Therefore, it may be worthwhile to revisit selected applications to determine if the presented techniques can be used to obtain stronger results.

A fundamental open question is whether a lower communication complexity for large messages, i.e., a lower overhead factor, can yet be attained, ideally without increasing the time complexity substantially. We conjecture that the lower bound of $3/2$, which has been shown for a specific class of algorithms, holds more generally, i.e., with one or both of the imposed restrictions lifted.
Deriving a tight bound for the overhead factor is another important avenue for future research.

\bibliographystyle{ACM-Reference-Format}
\bibliography{reliable_broadcast}

\end{document}